\documentclass[natbib]{svjour3}

\usepackage[margin=1.9cm]{geometry}

\usepackage{amsmath,graphicx,amssymb,fancyhdr,amsthm,enumerate,textcomp}

\usepackage[usenames]{color}
\usepackage{enumitem}

\newtheorem{thm}{Theorem}

\usepackage{color}

\newcommand\commentout[1]{} 
\usepackage{soul} 

\def\beq{\begin{eqnarray}}
\def\eeq{\end{eqnarray}}

\numberwithin{equation}{section}

\begin{document}
\title{Cartan Invariants and Event Horizon Detection}

\author{D. D. McNutt \and M.A.H. MacCallum \and D. Gregoris \and A. Forget \and A. A. Coley \and P. C. Chavy-Waddy \and D. Brooks}

\institute{D. D. McNutt \at Faculty of Science and Technology, University
 of Stavanger, N-4036 Stavanger, Norway
\email{david.d.mcnutt@uis.no}
 \and M. A. H. MacCallum \at  School of Mathematical Sciences, Queen
 Mary University of London, Mile End Road, London E1 4NS
\email{m.a.h.maccallum@qmul.ac.uk}
 \and D. Brooks \and P. C. Chavy-Waddy \and A. A. Coley \and
  A. Forget \and D. Gregoris \at
 Department of Mathematics and Statistics, Dalhousie University, 
 Halifax, Nova Scotia, Canada B3H 3J5 \email{dario.a.brooks@dal.ca,
   pl859327@dal.ca, aac@mathstat.dal.ca, Adam.AL.Forget@dal.ca,
   danielegregoris@libero.it}
  }

\date{}
\maketitle   
  
%




\begin{abstract}

We show that it is possible to locate the event horizons of a black hole (in arbitrary dimensions) as the zeros of certain Cartan invariants.  This approach accounts for the recent results on the detection of stationary horizons using scalar polynomial curvature invariants, and improves upon them since the proposed method is computationally less expensive.  As an application, we  produce Cartan invariants that locate the event horizons for various exact four-dimensional and five-dimensional stationary, asymptotically flat (or (anti) de Sitter) black hole solutions and compare the Cartan invariants with the corresponding scalar curvature invariants that detect the event horizon. In particular, for each of the four-dimensional examples we express the scalar polynomial curvature invariants introduced by Abdelqader and Lake in terms of the Cartan invariants and show a direct relationship between the  scalar polynomial curvature invariants and the Cartan invariants that detect the horizon. 
\end{abstract}

\section{Introduction}

General Relativity predicts the existence of singularities hidden by a horizon \citep{GRrefs}. Remarkably, this was noticed a year after the appearance of the Einstein field equations when Schwarzschild published a solution describing an isolated non-rotating massive object. Additional exact solutions to Einstein field equations have been found which exhibit this property \citep{Stephani}.  Improvements in astrophysical observations have allowed black holes to be distinguished from other highly massive objects like neutron stars in our universe suggesting the physical relevance of such general relativistic metrics \citep{astro1,astro1b, astro2, astro3}.

Naively speaking one can regard a black hole as a region of spacetime from which nothing can escape, i.e., after crossing the horizon towards the singularity, a photon can never escape to asymptotic infinity. While this captures the basic property of black holes it is clearly unsatisfactory in general relativity, which was constructed as a local theory, and the definition of an event horizon requires global information on the entire spacetime \citep{choq1, choq2}. Due to the contradictory nature of these two facts it is desirable to find alternative definitions or characterizations of black hole horizons that are quasi-local. For example, a local characterization of the horizon of a black hole is necessary in the numerical study of the evolution of configurations of many black holes. 

At this time, only approximate localizations are possible, such as considering the event horizon as a marginally outer trapped surface, a minimal surface, a Killing horizon or an apparent horizon \citep{Ashtekar, booth} which are also foliation dependent. Recently it was shown that specific combinations of the scalar polynomial curvature invariants (SPIs) (see next sections for their definition) vanish on the horizon of a stationary black hole. This provides a local technique for the localization of the event horizon, and an extension of \cite{PRM1993} allowing for the extraction of information about the mass, angular momentum and electric charge of a black hole  \citep{AbdelqaderLake2015, PageShoom2015}.

In this paper we will show that it is possible to locate the horizon of any stationary asymptotically flat (or (anti) de Sitter)  black hole using Cartan invariants. While both the SPIs and Cartan invariants are foliation independent, the Cartan invariants have two important advantages over SPIs: they are linear in terms of the components of the curvature tensor instead of quadratic or higher degree terms, and it is possible to construct from the Cartan invariants suitable invariants that vanish on the horizon and nowhere else, eliminating the problem of SPIs detecting surfaces outside of the horizon.
 
Using the Cartan-Karlhede algorithm (known as the Karlhede algorithm in 4D) \citep{ref1, ref2a, DAM2016} we briefly discuss the classification of metrics as a necessary step for  the computation of the Cartan invariants. We apply our method to four-dimensional (4D) black hole solutions and the less studied five dimensional (5D) black hole solutions \citep{string1, string2}. Finally, we compute SPIs for the 4D and 5D examples using the results of \cite{PageShoom2015}, and in 4D we show how the Cartan invariants are related to the SPIs: thereby the rather complicated expressions used for the SPIs in previous work are shown to have simpler forms.

\section{Horizon Detection with Scalar Polynomial Curvature invariants}

In this section we will review some basic properties of the SPIs that will be useful in the applications discussed in this paper.
In 1869, Christoffel showed that any scalar function on a $n$-dimensional Riemannian (or pseudo-Riemannian) manifold $(\textbf{M},g_{ab})$ constructed from the metric $g_{ab}$ must be a function of $R_{abcd}, R_{abcd;e}$ and  higher order covariant derivatives \citep{MacCallum2015}.  Due to the nature of SPIs, they are one of the conceptually simplest of such scalar functions.  The SPIs of a given spacetime metric, $g_{ab}$, are the set of functions generated by operations on (contractions of) the curvature tensors, and their covariant derivatives, such as \beq R_{ab}R^{ab}, C_{abcd}C^{abef}C_{ef}^{~~cd}, R_{ab;c}R^{ab;c},
C_{abcd;e}C^{abcd;e}. \eeq  We denote by $\mathcal{I} = \{R, R_{ab}R^{ab}, C_{abcd}C^{abcd}, \dots \}$ \noindent  the set of SPIs of $\textbf{M}$. Some basic examples, denoted by $I_1, \dots, I_7$ \citep{AbdelqaderLake2015,PageShoom2015}, are the following:

\begin{eqnarray}
\begin{aligned}\label{IInvariant}
& I_1 = C^{abcd}C_{abcd} , \qquad 
I_2 = C^{*abcd}C_{abcd} , \qquad 
I_3 = C^{abcd;e}C_{abcd;e} ,& \\
&I_4 = C^{*abcd;e}C_{abcd;e} , \qquad
I_5 = (I_1)_{;a}(I_1)^{;a} , \qquad
I_6 = (I_2)_{;a}(I_2)^{;a} , \qquad 
I_7 = (I_1)_{;a}(I_2)^{;a}, &
\end{aligned}
\end{eqnarray}

\noindent where $C_{abcd}$ is the Weyl tensor and $C^{*}_{abcd}$ is its dual and a semicolon denotes covariant differentiation. \\
We stress that the maximum number of functionally independent and maximum algebraically independent SPIs are in general different, the former being {\it at most} $n$, while the latter is 
\begin{equation}
N(n,p) = 
	\begin{cases}
		0 & \quad \text{if } p = 0 \text{ or } 1 \\
		\frac{n[(n+1)][(n+p)!]}{2n!p!} - \frac{(n+p+1)!}{(n-1)!(p+1)!} + n & \quad \text{if } p 			\geq 2 \\
	\end{cases}
\end{equation}
where $p$ denotes the order of differentiation of the metric tensor components.

Black hole horizon detection was remarked upon by \cite{KLA1982}, where the invariant $R_{abcd;e}R^{abcd;e}$ was shown to detect horizons for several type D solutions. However in the case of the Kerr horizon, it detected the stationary limit, and not the outer horizon itself. This was first noted by Skea in his  doctoral thesis \citep{Ske86} where it was observed that $R_{abcd;e}R^{abcd;e}$ did not provide an adequate test for horizons. More recently  \cite{AbdelqaderLake2015} examined a collection of invariants from which they determined physical properties of spacetimes around rotating black holes, including the detection of the horizons. These invariants are constructed from SPIs (note that being in vacuum we do not distinguish between Riemann and Weyl tensors):

\beq
\label{Q}
 Q_1 =  \frac{(I_1^2 - I_2^2)(I_5 - I_6)+4I_1 I_2 I_7}{3\sqrt{3}(I_1^2+I_2^2)^\frac{9}{4}}  , \quad 
Q_2 =  \frac{I_5 I_6 - I_7^2}{27 (I_1^2 + I_2^2)^\frac{5}{2}}  , \quad 
Q_3 =  \frac{I_5 + I_6}{6\sqrt{3} (I_1^2 + I_2^2)^\frac{5}{4}}, 
\eeq

\noindent where $I_1$ to $I_7$ are given by (\ref{IInvariant}).  From the dimensionless invariants $Q_1$, $Q_2$ and $Q_3$ one can read off the physical properties of the Kerr metric since they locate the horizon and ergosurface in an algebraic manner. With this information two approaches were provided to compute the angular momentum and mass of the black hole, one global and the other local. To determine the mass and angular momentum in the global approach, the area of the horizon and of the ergosurface must be calculated, requiring that these two surfaces must be located. 
The local method, which makes use of (\ref{IInvariant}) alone,  does not require knowledge of the location of the black hole or its event horizon. By knowing the forms of the invariants $I_1,...,I_7$ the mass and angular momentum can be expressed as functions in terms of these invariants. The derivation outlined by  \cite{AbdelqaderLake2015} is not unique, and the authors noted that the presented approach was the simplest found through experimentation. 

The relationship between (\ref{Q}) and (\ref{IInvariant}) has been expanded by \cite{PageShoom2015} in which the authors introduce a general approach to determine the location of the event horizon and ergosurface for the Kerr metric. More generally, this method will give the exact location of any horizon for a stationary black hole, although it is believed that it will be able to determine the approximate location for any nearly stationary horizon. This technique relies on the fact that the squared norm of the wedge product of $n$ gradients of functionally independent local smooth curvature invariants will always vanish on the horizon of any stationary black hole, where $n$ is the local cohomogeneity of the metric, which is defined as the codimension of the maximal dimensional orbits of the isometry group of the local metric. Their results can be summarized by the following theorem:

\begin{thm} \label{PSthrm}
For a spacetime of local cohomogeneity $n$ that contains a stationary horizon (a null hypersurface that is orthogonal to a Killing vector field that is null there and hence lies within the hypersurface and is its null generator) and which has $n$ independent SPIs $S^{(i)}$ whose gradients are well-defined there, the $n$-form wedge product 
\beq W = dS^{(1)} \wedge ... \wedge dS^{(n)} \nonumber \eeq
\noindent has zero squared norm on the horizon, 

\beq ||W||^2 = \frac{1}{n!} \delta^{\alpha_1,...,\alpha_n}_{\beta_1,...,\beta_n} g^{\beta_1 \gamma_1} ... g^{\beta_n \gamma_n} \times S^{(1)}_{;\alpha_1}...S^{(n)}_{;\alpha_n} S^{(1)}_{;\gamma_1}...S^{(n)}_{;\gamma_n} = 0, \nonumber \eeq
\noindent where the permutation tensor $\delta^{\alpha_1,...,\alpha_n}_{\beta_1,...,\beta_n}$ is $+1$ or $-1$ if $\alpha_1,...,\alpha_n$ is an even or odd permutation of $\beta_1,...,\beta_n$ respectively, and is zero otherwise. 
\end{thm}

In general the set $\mathcal{I}$ is not sufficient to  locally distinguish the manifold $\textbf{M}$ as it is possible that two different metrics can have the same set $\mathcal{I}$. In the particular case in which it \emph{is} fully characterized by its SPIs the spacetime is said to be $\mathcal{I}$-non-degenerate \citep{ref6}. If a spacetime metric is of Ricci type I, Weyl type I, or Riemann type I/G (relative to the alignment classification, which is reviewed in section \ref{CKalg5D}), then the metric is $\mathcal{I}$-non-degenerate. Moreover, in the case that the metric is not $\mathcal{I}$-non-degenerate, then it is necessarily contained in the Kundt class or is locally homogeneous \citep{ref6}. We note that all of the black hole metrics considered in this paper are $\mathcal{I}$-non-degenerate.

\section{The Cartan-Karlhede Method for Determining Local Equivalence of Spacetimes}

The method for testing geometric equivalence due to \'Elie Cartan \citep{Stephani} was developed with the aim of determining the equivalence of geometric objects under a diffeomorphism. The method we employ in our paper is a specialized form applicable to sets of differential forms defined on differentiable manifolds under appropriate transformation groups. In Riemannian geometry, the formal statement of this problem is:

\bigskip
\noindent{\it Let $({\bf M}, g)$ and $({\bf \bar{M}},\bar{g})$ be two $m$-dimensional Riemannian manifolds. We say that $g$ and $\bar{g}$ are equivalent as Riemannian metrics if $$\Phi^*(\bar{g})=g,$$ for a locally-defined diffeomorphism $\Phi:\textbf{M}\rightarrow \bar{\textbf{M}}$. When does such a diffeomorphism exist?}
\bigskip


To relate two apparently different metrics, we look to the coordinate neighbourhoods defined on the manifold at each point and examine the frame bundle on the manifold. If the metrics on two neighbourhoods are equivalent, the frame bundles derived from them are (locally) identical.  Now, the frame bundle on each manifold possesses uniquely-defined one-form fields $\{\boldsymbol{\omega}^a, {\bf \Gamma}^a_{~ b}\}$ such that 
$${\bf \Gamma}^a_{~ b}=\Gamma^a_{ \hspace{2 mm}bc}\boldsymbol{\omega}^c, \hspace{5mm} \Gamma^a_{\hspace{2 mm}bc}=\langle \boldsymbol{\omega}^a,\nabla_c{\bf e}_b\rangle,$$
where $\{{\bf e}_a\}$ is the corresponding basis of tangent vectors and $\{\boldsymbol{\omega}^a\}$ the basis of one-forms. In local coordinates $\{x^1,\dots, x^m\}$,
$${\bf e}_a=e_a^i\frac{\partial}{\partial x^i}, \hspace{5mm} \boldsymbol{\omega}^a=\omega^a_i dx^i.$$
Thus, the basis one-form fields on the frame bundle ($\Gamma, \omega$) defined on each frame bundle must also be identical for equivalent metrics. In the case of spacetimes, the Cartan structure equations imply: 
\begin{equation}
d\boldsymbol{\omega}^a=-\boldsymbol{\Gamma}^a_{~ b}\wedge\boldsymbol{\omega}^b, \label{Cartan1}
\end{equation}
\begin{equation}
d\boldsymbol{\Gamma}^a_{~ b}+\boldsymbol{\Gamma}^a_{~c}\wedge\boldsymbol{\Gamma}^c_{~ b}=\boldsymbol{\Theta}^a_{~b},~~{\bf \Theta}^a_{~b} = R^a_{~bcd} \boldsymbol{\omega}^c \wedge \boldsymbol{\omega}^d.\label{Cartan2}
\end{equation}
Hence the Cartan structure equations show that the above implies the components of the curvature on the frame bundle must also be equatable.

However, the equability  of \eqref{Cartan1} and \eqref{Cartan2} on the two neighbourhoods are necessary conditions, and not sufficient for the local identification of differentiable manifolds. Cartan showed that sufficient conditions are obtained by taking repeated exterior derivatives, starting with $d \Gamma^a_{~b}$, until no new functionally independent quantity arises; if at any step of differentiation no such quantity arises, the process terminates, as any further derivatives depend on the quantities previously obtained. Consequently, the relations between dependent and independent invariants must be the same in coordinate neighbourhoods of points in both manifolds in order for them to be equivalent. The number of $k$ functionally independent quantities obtained (called the {\it rank}) is at most the dimension $m$ of the manifold, and so the process necessarily terminates in a finite number of steps; if it turns out that $k<m$ this is due to the presence of symmetries.

In the case of spacetimes, the equation
$$dR_{abcd}=R_{abcd;e}\boldsymbol{\omega}^e+R_{ebcd}\boldsymbol{\Gamma}^e_a+R_{aecd}\boldsymbol{\Gamma}^e_b+R_{abed}\boldsymbol{\Gamma}^e_c+R_{abce}\boldsymbol{\Gamma}^e_d$$
shows that repeated exterior differentiation is equivalent to repeatedly taking covariant derivatives of the Riemann tensor $R_{abcd}$. A metric can consequently be uniquely (locally) characterized by its Riemann tensor and a finite number of its covariant derivatives, regarded as functions on the frame bundle of the manifold. If we use $R^q$ to denote the set $\{R_{abcd}, R_{abcd;f}\dots,R_{abcd;f_1f_2\dots f_q}\}$ of the components of the Riemann tensor and its covariant derivatives up to the $q$th order, then if $p$ is the last derivative at which a new functionally independent quantity arises (the {\it order}), we must compute $R^{p+1}$. If $k$ is the number of functionally independent invariants on the frame bundle in a maximal set, we denote the invariants by $I^{\alpha}$, $\alpha=1,\dots ,k$.


The main idea of this method is to reduce the frame bundle to the smallest possible dimension at each step by casting the curvature and its covariant derivatives into a canonical form and only permitting those frame changes which preserve that canonical form. The frames we will employ in 4D are the so-called null tetrads, i.e. a set of four complex vectors $\{l^a, n^a, m^a, {\bar m^a}\}$ such that $l^al_a=n^a n_a=m^am_a= {\bar m^a}{\bar m_a}=0$ and $l_a n^a=1=m^a {\bar m_a}$ and where a bar denotes complex conjugate. In terms of this complex null tetrad the metric is 
\beq
ds^2=-2l_{ (a }n_{ b) }+2m_{ (a }{\bar m  }_{b)},
\eeq
\noindent where round parentheses denote symmetrization.

Since there are many solutions to Einstein's equations describing vacuum spacetimes and few which admit conformally flat geometries, we usually begin by putting the Weyl tensor into the appropriate normal form (See \cite{Stephani} section 4.2, table 4.2.), and then using any residual frame freedom to put the Ricci tensor $R_{ab}=R^c_{~acb}$ into canonical form, if possible. The curvature components in this tetrad are the first set of invariants required. We then calculate the first (covariant) derivatives of the curvature and use them to further fix the tetrad, if necessary. This is repeated for higher derivatives until the stopping conditions are met for the algorithm, which will be discussed below.

The Cartan-Karlhede algorithm that we will use in the next section is \citep{Mac86}:
\begin{enumerate}
\item Set the order of differentiation $q$ to 0.
\item Calculate the derivatives of the Riemann tensor up to the $q$th order.
\item Find the canonical form of the Riemann tensor and its covariant derivatives.
\item Fix the frame as much as possible using this canonical form, and note the residual frame freedom (the group of allowed transformations is the {\it linear isotropy group $H_q$}). The dimension of $H_q$ is the dimension of the remaining {\it vertical} freedom of the frame bundle.
\item Find the number $t_q$ of independent functions of space-time position in the components of the Riemann tensor and its covariant derivatives, in the canonical form. This tells us the remaining {\it horizontal} freedom.
\item If the isotropy group and number of independent functions are the same as in the previous step, let $p+1=q$, and the algorithm terminates; if they differ (or if $q=0$), increase $q$ by 1 and go to step 2. 
\end{enumerate}

The nonzero components of $R_{abcd}$ and its covariant
derivatives are referred to as {\it Cartan invariants}: a statement of
the minimal set required, taking Bianchi and Ricci identities into
account, was given by \cite{MacAma86}. We will refer to the invariants constructed from, or equal to, Cartan invariants of any order as {\it extended invariants}. Thus  for sufficiently smooth metrics, a result of the test of equivalence gives sets of scalars providing a unique local geometric characterization, as the $D$-dimensional space-time is then characterized by the canonical form used, the {two discrete sequences arising from the } successive isotropy groups and the independent function counts, and the values of the (nonzero) Cartan invariants. As there are $t_p$ essential space-time coordinates, the remaining $D-t_p$ are ignorable, and so the dimension of the isotropy group of the space-time will be $s=\dim(H_p)$, and the isometry group has dimension $r=s+D-t_p$. 

Theorem \ref{PSthrm} can be readily generalized to the set of Cartan invariants arising from the Cartan-Karlhede algorithm:

\begin{thm} \label{DPSthrm}
For a spacetime of local cohomogeneity $n$ that contains a stationary horizon and which has $n$ independent Cartan invariants $C^{(i)}$ whose gradients are well-defined there, the $n$-form wedge product 
\beq W = dC^{(1)} \wedge ... \wedge dC^{(n)} \nonumber \eeq
\noindent has zero squared norm on the horizon, 

\beq ||W||^2 = \frac{1}{n!} \delta^{\alpha_1,...,\alpha_n}_{\beta_1,...,\beta_n} g^{\beta_1 \gamma_1} ... g^{\beta_n \gamma_n} \times C^{(1)}_{;\alpha_1}...C^{(n)}_{;\alpha_n} C^{(1)}_{;\gamma_1}...C^{(n)}_{;\gamma_n} = 0, \nonumber \eeq
\noindent where the permutation tensor $\delta^{\alpha_1,...,\alpha_n}_{\beta_1,...,\beta_n}$ is $+1$ or $-1$ if $\alpha_1,...,\alpha_n$ is an even or odd permutation of $\beta_1,...,\beta_n$ respectively, and is zero otherwise. 
\end{thm} 

\begin{proof}
The number of functionally independent invariants at the end of the
algorithm, $t_p$,  is directly related to the dimension of the local
cohomogeneity. To see this, we note that the dimension of the isometry
group is given by $r = D-t_p + dim(H_p)$ where $H_p$ is the dimension
of the isotropy group of the curvature tensor and all its covariant derivatives. However, the maximal dimensional orbits of the isometry group will be given by $r-dim(H_p)$ = $D-t_p$, since this is the quotient of the Lie group of Killing vectors by the isotropy group, and therefore $n = D - r + dim(H_p) = t_p$. Using $n$ functionally independent Cartan invariants, the proof carries forward in a similar manner to the proof of theorem  \ref{PSthrm} in \cite{PageShoom2015}. 
\end{proof}

\noindent Alternatively, we can use the first order Cartan invariants (those arising from the covariant derivative of the Riemann tensor) to produce new invariants that detect the stationary horizons. These invariants will be much simpler than the SPIs.

\subsection{The Cartan-Karlhede Algorithm in Five Dimensions} \label{CKalg5D}

We would like to apply the Cartan-Karlhede algorithm to determine a set of Cartan invariants which detect the stationary horizon for the 5D black hole metrics. This can be achieved in arbitrary dimension by examining the $q$th covariant derivative of the Weyl and Ricci tensor at iteration, $q$, and using the frame transformations to transform the $q$th covariant derivative of the Weyl tensor and Ricci tensor into some canonical form, if possible.

In 5D, relative to the half-null frame  with $n_al^a=1$, $n_an^a=l_al^a=0$ and $m^{(i)}_am^{(j)}_b=\delta^{ij}$ in terms of which the metric can be written as $g_{ab}=2 l_{(a}n_{b)}+\delta^j_im^{(i)}_am^{(j)}_b$, the local Lorentz transformations are generated by combining the following frame transformations \citep{ref3,ref4}:

\begin{subequations}\label{Lorentz5D}
\begin{enumerate}
  \item Null rotations about $l$:
  \beq
 \hat l= l, \quad \hat n= n+z_i m^i-\frac{1}{2} z_i z^i l, \quad \hat m_i= m_i-z_i l 
\eeq
  \item Null rotations about $n$:
  \beq
\hat l= l+y_i m^i-\frac{1}{2} y_i y^i n, \quad \hat n= n, \quad \hat m_i= m_i-y_i n
\eeq
  \item Spins:
\beq
\hat l=l , \quad \hat n=n , \quad \hat m_i =X_i^j m_j
\eeq
  \item Boost:
\beq
\label{boost} 
\hat l= \lambda l , \quad \hat n=\lambda^{-1} n , \quad \hat m_i = m_i\,,
\eeq

\end{enumerate}
\end{subequations}

\noindent where $X_i^j$ denotes the usual rotation matrices 
for rotations about the axes $m_2$, $m_3$, $m_4$ respectively. We stress that the quantities $z_i=z_i(x^a)$, $\theta=\theta(x^a)$ and $\lambda=\lambda(x^a)$ 
depend on the coordinates. We also note that the Lorentz transformations in 5D have 10 parameters.

For dimension $D > 4$, we no longer have the usual spinor approach to simplify calculations, and the
4D algebraic classifications of the Weyl and Ricci tensors are no longer applicable.
Instead, we consider the boost weight decomposition \citep{CH2010,CH2011,CHOL2012} to classify the curvature tensor. Relative to the basis $\{ \theta^a\} = \{ n, \ell, m^i\}$, the components of an
arbitrary tensor of rank $p$ transform under the
boost \eqref{boost} by:
\beq T'_{a_1 a_2...a_p} = \lambda^{b_{a_1 a_2 ... a_p}} T_{a_1 a_2 ... a_p},~~ b_{a_1 a_2...a_p} = \sum_{i=1}^p(\delta_{a_i 0} - \delta_{a_i 1})  \eeq
\noindent where $\delta_{ab}$ denotes the Kronecker delta symbol. This
quantity is called the {\it boost weight} (b.w) of the frame component
$T_{a_1 a_2 ... a_p}$. This approach, called the {\it alignment
  classification}, relies on the fact that the frame basis written as
a null basis transforms in a simple manner under a boost given by
(\ref{boost}) and that this identifies null directions relative to
which the Weyl tensor has components of a particular
b.w. configuration, called Weyl aligned null directions
(WANDs). Typically, we must use null rotations to identify the WANDs for a given tensor.

We define the {\em boost order} of $T_{a_1 a_2 ... a_p}$ as the maximal b.w. of its non-vanishing components relative to the frame. As this integer is invariant under the group of Lorentz transformations that fix the null direction $[\ell]$, it is a function of $[\ell]$ only, and will be denoted by $b_T([\ell])$. We introduce another integer, $B_T = \max_{\ell} b_T([\ell])$, which is entirely dependent on the form of the tensor. For a generic $\ell$ the Weyl and Ricci tensors have boost order $b_R([\ell]) = b_C([\ell]) = 2$, and so $B_R=B_C = 2$. If a null direction $[\ell]$ exists for which $b_T([\ell]) \leq B_T-1$, it is said to be a {\em $T$ aligned null direction} of {\em alignment order}: $B_T-1-b_T([\ell])$. As an example, for a WAND, the alignment order can be $0,1,2,3$. The alignment order can be related to another integer invariant, 
\beq \zeta \equiv \min_{\ell} b_C([\ell]), \nonumber \eeq
\noindent which is a pointwise invariant of the spacetime defining the 
(Weyl) primary or principal alignment type $2-\zeta$ at p. If $\zeta = 2,1,0,1$ or $-2$ this type is denoted by {\bf G}, {\bf I}, {\bf II}, {\bf III} or {\bf N} respectively. If there is more than one WAND in the type {\bf II} case, then this is denoted as {\bf D}. This classification can also be applied to the Ricci tensor since $B_C = B_R = 2$. 

This classification reproduces the Petrov and Segre classifications in 4D, and also leads to a coarse classification in higher dimensions; for example, an algebraic classification of all higher dimensional Kundt metrics \citep{PS2013}. In 5D this classification can be made finer by considering the spin group which is isomorphic to $O(3)$ and acts on the null frame according to (\ref{Lorentz5D}). The details of this approach are expanded upon in \cite{CHOL2012}. There is a fundamental difficulty with applying the alignment classification, as it relies on solving degree four multivariate polynomials, with more than two variables, to determine the WANDs. The exact solutions to such polynomial equations are difficult to compute and hence the ability to determine the WANDs in dimension higher than four is not always feasible in practice. Assuming a theory of approximate equivalence could be developed, numerical root solving could be implemented to resolve this issue.

\section{Applications in 4D}\label{WE}

In this section we apply the Cartan equivalence method for the classification of 4D solutions of the Einstein equations describing stationary, asymptotically flat (or (anti) de Sitter) black holes to compute Cartan invariants that are capable of identifying the horizons which correspond to the positive b.w. components of the covariant derivatives of the Weyl and Ricci tensor; this follows in 4D from the fact that a Killing horizon is a special case of a weakly isolated horizon  \citep{ADA2017,AD2017}. We will then relate these Cartan invariants to the SPIs using the Newman-Penrose (NP) formalism. Due to the relationship between the SPIs $I_1$ and $I_2$ and the Cartan invariant $\Psi_2$ each of the SPIs generated by Theorem \ref{PSthrm} may be considered as extended Cartan invariants produced by Theorem \ref{DPSthrm}. 

We note that some of the examples we consider here will be contained as special cases in others. For example, the Kerr solution is a special case of the Kerr-NUT-(Anti)-de Sitter solution and the Kerr-Newman solution,  and similarly the Reissner-Nordstr{\"o}m solution is contained in the Kerr-Newman solution. Our intention in giving these as separate cases is to provide examples in different coordinate systems and illustrate the structure of the Cartan invariants.

\subsection{Kerr metric}

The 4D Kerr metric in Boyer-Lindquist coordinates is given by the line element 
\small
\beq \begin{aligned} 
ds^2 &=& -\frac{Q-a^2 \sin^2 \theta}{R^2} dt^2 - \frac{2a \sin^2 \theta (r^2 + a^2 - Q) }{R^2} dt d \phi + \frac{((r^2+a^2)^2 - Q a^2 \sin^2 \theta) \sin^2 \theta}{R^2} d\phi^2 + \frac{R^2}{Q} dr^2 + R^2 d\theta^2 ,
\end{aligned}  \label{KerrMtrc} \eeq
\beq Q(r) = r^2 + a^2 -2Mr,~~ R(r,\theta) = \sqrt{r^2 + a^2 \cos^2 \theta}. & \label{KerrFns} \eeq
\normalsize

To start the Cartan-Karlhede algorithm, we employ the following null coframe:
\begin{eqnarray}
\ell&=&dt+\left(\frac{R^2}{Q}\right)dr+a \sin^2\theta d\phi, \\
n&=&\left(\frac{Q}{2R^2}\right)dt-\frac{1}{2}dr+\left(\frac{a \sin^2\theta Q }{2R^2}\right)d\phi, \\
m&=&\left(-\frac{i\sqrt{2}a \sin\theta}{2\left(r+i a \cos\theta\right)}\right)dt+\left(\frac{\sqrt{2} R^2 }{2\left(r+i a \cos\theta\right)}\right)d\theta+\left(\frac{i\sqrt{2}(a^2+r^2)\sin\theta}{2\left(r+i a \cos\theta\right)}\right)d\phi, \\
\overline{m}&=&\left(-\frac{i\sqrt{2}a \sin\theta}{2\left(i a \cos\theta-r\right)}\right)dt-\left(\frac{\sqrt{2} R^2}{2\left(i a \cos\theta-r\right)}\right)d\theta+\left(\frac{i\sqrt{2}(a^2+r^2)\sin\theta}{2\left(i a \cos\theta-r\right)}\right)d\phi.
\end{eqnarray}

\noindent To calculate the Riemann tensor at zeroth order we will use the NP formalism. We find the only non-vanishing curvature scalar at zeroth order to be:
\beq
\Psi_2=\frac{iM}{\left(a \cos\theta+i r\right)^3}.
\eeq
Any null rotation will ruin the form of the Riemann tensor, while spins and boosts do not effect $\Psi_2$, and thus leave the Weyl tensor invariant. The dimension of the isotropy group has been reduced from six to two at zeroth order. From $\Psi_2$, we obtain an even simpler invariant $ C_0=\left(i \frac{1}{\Psi_2}\right)^{\frac{1}{3}};$ the real and imaginary parts define two functionally independent zeroth order invariants:
\beq Re(C_0)=\frac{a \cos\theta}{M^{\frac{1}{3}}},\hspace{3mm} Im(C_0)=\frac{r}{M^{\frac{1}{3}}}.\eeq
The zeroth iteration of the Cartan-Karlhede algorithm concludes with $\dim(H_0)=2$, $t_0=2$.

To begin the first iteration of the Cartan-Karlhede algorithm, we compute the first covariant derivative of the Weyl tensor (or the Weyl spinor here, as we have chosen to work with the spinor formalism \citep{Stephani}). From the symmetrized first covariant derivative of the Weyl spinor has the following components:
\beq 
\hspace*{-2em}\nabla \Psi_{20'} = 3(D \Psi_2 +2\rho\Psi_2)/5,~
\nabla \Psi_{21'} = 3(\delta\Psi_2+2 \tau \Psi_2)/5 ,~
\nabla \Psi_{30'} = 3(\bar{\delta} \Psi_2 -2\pi\Psi_2)/5,~
\nabla \Psi_{31'} = 3(\Delta\Psi_2 -2 \mu \Psi_2)/5.
\label{NablaPsi0}
  \eeq
Since we have that $\rho$, $\mu$, $\tau$, and $\pi$ are all non-vanishing, the Kerr metric belongs to the family of type D vacuum spacetimes in the third case identified by \cite{ref1}; we thus may fix the boost and spin to some desired value. We employ the canonical choice for such a Petrov type D metric, $\nabla\Psi_{31'}=-\nabla\Psi_{20'}$, implying that $\rho = \mu$ and additionally impose that $\tau = \pi$ using the remaining spin:

\beq & \rho = \mu = \frac{i\sqrt{Q}}{\sqrt{2} R (ir + a \cos \theta) }, & \label{rhoKerr} \\
& \tau = \pi = \frac{-i a \sin\theta }{\sqrt{2} R (ir+a\cos\theta)}. & \label{tauKerr} \eeq

\noindent No new functionally independent invariants have been introduced at this iteration. The first iteration of the Cartan-Karlhede algorithm therefore concludes with $\dim(H_1)=0$ and $t_1=2$. Although $t_0=t_1=2$, we have that $2=\dim(H_0)\neq\dim(H_1)=0$ and so we must continue the algorithm.

The second iteration of the Cartan-Karlhede algorithm begins by computing the second covariant derivative of the Weyl spinor.  While this can be computed in compact form using the GHP formalism and the formulae $(4.3a')-(4.3i')$ of \citep{ref1}, we will omit the details. No functionally independent invariants appear from the second covariant derivative. Thus, $t_1=t_2=2$ and $\dim(H_1)=\dim(H_2)=0$; the Cartan-Karlhede algorithm terminates after the second iteration, in agreement with \cite{Aman1984}. In the subsequent examples, we will omit  details of the final iteration of the Cartan-Karlhede algorithm for brevity.

Using the Cartan invariants we can construct scalar invariants that can be used to detect both the horizon and the ergosurface. 
For vacuum Petrov type D metrics the Bianchi identities give
\beq &D \Psi_2 = 3 \rho \Psi_2,~~ \Delta \Psi_2 = - 3 \mu \Psi_2 , & \nonumber \\ 
& \delta \Psi_2 = 3\tau \Psi_2,~~\bar{\delta}
\Psi_2 = -3\pi \Psi_2. & \label{BianchiVacuum} \eeq
\noindent Applying the Bianchi identities to \eqref{NablaPsi0}, the extended Cartan invariant 
\beq \frac{\nabla \Psi_{20'}}{\Psi_2} \propto \rho, \eeq
\noindent will vanish on the horizon, and nowhere else due to its coordinate expression \eqref{rhoKerr}. This invariant is relevant due to its geometric meaning as noted by \cite{MacCallum2006}, but it is not unique, since one could use another combination of Cartan invariants up to first order. Noting that $Q_1$ detects the ergosurface, we would like a corresponding Cartan invariant that will do so. Consider the following extended Cartan invariant:

\beq \rho^2 - \tau^2 = \frac{-(Q - a^2 \sin^2 \theta) }{2 R^2 (ir + a\cos \theta)^2}, \label{ergokerr} \eeq

\noindent comparing with the $g_{tt}$ component of \eqref{KerrMtrc} shows that this will detect the ergosurface.

It is important to stress that this approach requires a particular invariantly defined choice of coframe, and that $\rho$ and $\rho^2 - \tau^2$ can be regarded as invariants only in the form  they take relative to the canonical frame.  However, it is possible to implement Theorem \ref{DPSthrm} to generate an extended Cartan invariant that detects the event horizon and make the choice of frame irrelevant. Working with $\Psi_2$ and its complex conjugate we find the following invariant:

\beq ||d\Psi_2 \wedge d\bar{\Psi}_2 ||^2 = \frac{-3^4 2^3  Q^2 a^2 M^4 \sin^2 \theta}{R^{20}}. \label{OurQ2Kerr} \eeq


%
%
%

\subsection{Reissner-Nordstr\"{o}m-(Anti)-de Sitter metric}

In this section we consider the Cartan invariants arising from the Cartan-Karlhede algorithm for the 4D Reissner-Nordstr\"{o}m-(Anti)-de Sitter metric describing a static but electrically charged black hole in presence of a cosmological constant \citep{Stephani}. In the system of coordinates ($t$, $r$, $\theta$, $\phi$) the metric is 

\beq
ds^2=-f(r)dt^2+\frac{dr^2}{f(r)}+r^2(d\theta^2+\sin^2\theta d\phi^2)\,
\eeq
where $f(r) = \left(1-\frac{2M}{r}+\frac{q^2}{r^2}-\frac{\Lambda r^2  }{ 3 }\right)$, and $M$ and $q$ denote respectively the mass and the electric charge of the black hole and $\Lambda$ the cosmological constant. 

To begin the Cartan-Karlhede algorithm for the Reissner-Nordstr{\"o}m-(anti)-de Sitter metric, we introduce the orthonormal frame:
\beq
e_{ 0 }=\frac{1  }{\sqrt{1-\frac{2M}{r}+\frac{q^2}{r^2} -\frac{\Lambda r^2  }{ 3 }  } }\partial_t, \quad
e_{ 1 }=\sqrt{1-\frac{2M}{r}+\frac{q^2}{r^2}-\frac{\Lambda r^2  }{ 3 }} \partial_r, \quad e_{ 3 }=\frac{1  }{  r  }\partial_\theta, \quad
e_{4 }=\frac{1  }{r\sin \theta   }\partial_\phi, 
\eeq
in terms of which we construct the  null tetrad:
\beq
l=\frac{ 1 }{ \sqrt{2  } }(e_{1  }+e_{2 }), \quad n=\frac{ 1 }{ \sqrt{2  } }(e_{ 1 }-e_{ 2 }), \quad m=\frac{ 1 }{ \sqrt{2  } }(e_{ 3 }+i e_{ 4}), \quad {\bar m  }=\frac{ 1 }{ \sqrt{2  } }(e_{3 }-i e_{4})\,.
\eeq

\noindent The nonzero curvature scalars are
\beq
\Phi_{ 11 }=\frac{ q^2 }{2r^4  }, \quad \Psi_2=\frac{ q^2-Mr }{ r^4 }, \quad \Lambda_{NP}=\frac{  \Lambda}{  6},
\eeq
\noindent where $R$ is the Ricci scalar.

Calculating the first covariant derivative of the Weyl and Ricci spinors are given by \eqref{NablaPsi0} and 
\begin{eqnarray}
\begin{aligned} 
\nabla\Phi_{11'} = 4(D\Phi_{11}+(\rho+\bar{\rho})\Phi_{11})/9,\\
\nabla\Phi_{12'} = 4(\delta\Phi_{11}+(\tau-\bar{\pi})\Phi_{11})/9,\\
\nabla\Phi_{22'} = 4(\Delta\Phi_{11}-(\mu+\bar{\mu})\Phi_{11})/9,
\end{aligned} \label{NablaPhi}
\end{eqnarray} 
\noindent where 
\beq & \rho = \mu = -\frac{1}{ \sqrt{2} r } \left(1-\frac{2M}{r}+\frac{q^2}{r^2}-\frac{\Lambda r^2  }{ 3 }\right)^\frac12, & \label{rhoRNADs} \\
& \tau = \pi = 0. \label{tauRNADs} \eeq

\noindent At first order we still have one functionally independent component and the boost is no longer in the isotropy group; thus $t_1=1$ and dim($H_1$)=1. The nonzero components are expressed in terms of the frame derivatives of $\Psi_2$, the nonzero Ricci spinor components, and the spin-coefficients $\mu, \rho, \pi$ and $\tau$. As the Cartan-Karlhede algorithm stops at second order \citep{Aman1984}, we can now refer to these quantities as Cartan invariants. 

For Petrov type D metrics in which $\Phi_{11}$ is the only nonzero matter term, the Bianchi identities give
\beq &D \Psi_2 = 3 \rho \Psi_2 + 2 \rho \Phi_{11},~~ \Delta \Psi_2 = - 3 \mu \Psi_2 - 2 \mu \Phi_{11}, & \nonumber \\ 
& \delta \Psi_2 = 3\tau \Psi_2 - 2 \tau \Phi_{11},~~\bar{\delta}
\Psi_2 = -3\pi \Psi_2 + 2 \pi \Phi_{11}, & \label{Bianchi} \\
& D\Phi_{11} = 2(\rho + \bar{\rho}) \Phi_{11},~~ \delta \Phi_{11} = 2(\tau - \bar{\pi}) \Phi_{11},~~ \Delta \Phi_{11} = -2(\mu + \bar{\mu}) \Phi_{11}.& \nonumber \eeq
\noindent Applying the Bianchi identities to \eqref{NablaPsi0}, we find that the extended Cartan invariant 
\beq \frac{\nabla \Psi_{20'}}{\Psi_2} \propto \rho, \eeq
\noindent will vanish on the horizon, and nowhere else, due to \eqref{rhoKerr}.

Previously it was noted that $I_3 = R^{abcd;e} R_{abcd;e}$ detects the horizon for the Reissner-Nordstr\"{o}m solution \citep{KLA1982}, which can be seen by direct calculation: 
\beq 
&&I_3=R_{ abcd;e }R^{abcd ;e }= -\frac{16 (15M^2r^2-36Mq^2r+22q^4) f(r) }{r^{ 10 } }. \eeq
\noindent Alternatively, we can apply Theorem \ref{DPSthrm} to generate an extended Cartan invariant that will detect the horizon. As the cohomogeneity of this solution is $n=1$, we may consider the norm of the exterior derivative of $\Psi_2$:

\beq || d \Psi_2 ||^2 = \frac{(3Mr-4q^2)^2 f(r)}{r^{10}}. \eeq

\subsection{Kerr-Newman metric}
It is worthwhile to ask if the invariants $Q_1, Q_2$ and $Q_3$ from (\ref{Q}) detect the event horizon and ergosurface for more general type D metrics, as for a non-vacuum solution. To check this we consider the Kerr-Newman solution given by the line element:

\beq  ds^2 = \frac{Q}{R^2} (dt - a\sin^2(\theta) d \phi)^2 - \frac{R^2}{Q}dr^2 -R^2 d\theta^2 - \frac{(r^2 +a^2)^2 \sin^2 (\theta)}{R^2} \left(d\phi - \frac{a}{r^2+a^2} dt \right)^2 \label{KerrNwmnMetric} \eeq
 where $\Lambda$ denotes the cosmological constant and $q$ denotes the electric charge, and the functions $Q$ and $R$ are:
\beq Q = r^2 - 2Mr + q^2 + a^2, \qquad  R = \sqrt{r^2 + a^2 \cos^2(\theta)}.  \eeq

\noindent The location of event horizons may be calculated from the zeros of $Q$. In the limit $q\to0$ we recover the line element for the Kerr metric.  

The expressions for the invariants $Q_1, Q_2$ and $Q_3$ are considerably larger and we will show in section \ref{SPIalaCartan} that $Q_1$   no longer detects the ergosurface for the Kerr Newman solution. While it is still possible to apply Theorem \ref{PSthrm} from \cite{PageShoom2015} to generate a SPI that will detect the horizon, relative to the coordinates, the calculation of this invariant will be lengthy.

We would like to see if a simpler invariant can be built out of first order Cartan invariants.  We begin the Cartan-Karlhede algorithm by computing the Weyl and Ricci spinors using the NP formalism. We define our orthonormal coframe:
\beq & e_0 = \frac{\sqrt{Q}}{R} dt - \frac{\sqrt{Q} a \sin (\theta) }{R} d\phi  , \quad e_1 = \frac{R}{\sqrt{Q}} dr  ,\quad e_2 = R d \theta , \quad   e_3 = \frac{(x^2+a^2) \sin (\theta)}{R} d \phi - \frac{a \sin (\theta)}{R} dt, & \eeq
in terms of which the full tetrad reads:
\beq
\ell=\frac{1}{\sqrt{2}}(e_0-e_1), \quad  n=\frac{1}{\sqrt{2}}(e_0+e_1), \quad m=\frac{1}{\sqrt{2}}(e_2+i e_3), \quad \overline{m}=\frac{1}{\sqrt{2}}(e_2-i e_3). \eeq

\noindent The nonzero curvature scalars are:
\beq
\Phi_{11}= \frac{q^2}{2 \left(r^2+a^2 \cos^2\theta\right)^2}, \quad
\Psi_2=\frac{i \left(Ma \cos\theta-i Mr+i q^2\right)}{\left(a \cos\theta-i r\right)\left(a \cos\theta+i r\right)^3}.
\eeq
Using $\Phi_{11}$ and (the magnitude of) $\Psi_2$, we construct the functionally independent invariants:
\beq
C_0=\frac{M^2a^2 \cos^2\theta+M^2r^2-2Mrq^2+q^4}{q^4}, \quad
C_1=\frac{r^2+a^2 \cos^2\theta}{q}.
\eeq
Thus, $t_0=2$. In addition, null rotations alter the form of $\Psi_2$ and $\Phi_{11}$, while boosts and spins leave both unchanged. We have therefore reduced the isotropy group from six dimensions to two at zeroth order, i.e. $\dim(H_0)=2$. 

We can now proceed to the first iteration of the Cartan-Karlhede algorithm, by calculating the first covariant derivative of the Weyl and Ricci spinors \citep{ref1,ref2a}, these are given by \eqref{NablaPsi0} and \eqref{NablaPhi}
\noindent where we may apply a boost and spin to set:
\beq & \rho = \mu = \frac{-i \sqrt{Q}}{\sqrt{2} R (ir + a\cos \theta)}, & \label{rhoKerrNewman} \\
& \tau = \pi = \frac{-i a \sin \theta }{\sqrt{2} R (ir + a\cos \theta)}. & \label{tauKerrNewman}\eeq

No new functionally independent Cartan invariants appear at first order. The remaining isotropy freedom is used up at first order by fixing both boosts and spins to be identity. It is known already \citep{Aman1984} that the Cartan-Karlhede algorithm concludes at the second iteration, since no new functionally independent invariants appear; $t_1 = t_2 = 2$ and dim $(H_1)$ = dim $(H_2)$ = 0. 

As an aside, we use components of the Weyl spinor and its first covariant derivative to construct the Cartan invariants that will detect the horizon and ergosurface. Applying the Bianchi identities \eqref{Bianchi} to \eqref{NablaPsi0}, the same extended Cartan invariant, 
\beq \frac{\nabla \Psi_{20'}}{\Psi_2} \propto \rho, \eeq
\noindent will vanish on the horizon, and nowhere else due to \eqref{rhoKerr}. Unlike $Q_1$ the extended Cartan invariant that detects the ergosurface is applicable to the Kerr-Newman solution since

\beq \rho^2 - \tau^2 = \frac{-(Q - a^2 \sin^2 \theta) }{2R^2 (ir + a\cos \theta)^2} \eeq
\noindent and comparing with the $g_{tt}$ component of \eqref{KerrNwmnMetric} shows that this will detect the ergosurface.

As an alternative, we apply Theorem \ref{DPSthrm} using the zeroth order Cartan invariants $\Psi_2$ and $\bar{\Psi}_2$, which gives the following extended Cartan invariant that will detect the horizon: 
\beq W =  || d \Psi_2 \wedge d \bar{\Psi}_2||^2 = \frac{- a^2 Q \sin^2(\theta)(9M^2 \cos^2 (\theta) a^2 + 9 r^2 M^2 -18rMq^2+8q^4)^2}{2(r^2+a^2\cos^2(\theta))^{12}} . \eeq

\subsection{Kerr-NUT-(Anti)-de Sitter metric}

The 4D Kerr-NUT-AdS metric is given by the line element \citep{PlebDem1976,Stephani,Griffiths2007}
\begin{eqnarray}
ds^2&=&\left(\frac{P-Q}{p^2+q^2}\right)dt^2+\left(\frac{Q  p^2+P  q^2}{p^2+q^2}\right)\left(dt\otimes dr+dr\otimes dt\right) \nonumber\\
&+&\left(\frac{P  q^4-Q  p^4}{p^2+q^2}\right)dr^2 +\left(\frac{p^2+q^2}{P}\right)dp^2+\left(\frac{p^2+q^2}{Q}\right)dq^2,
\end{eqnarray}
where $P\equiv P(p)$ and $Q\equiv Q(q)$ are fourth-degree polynomials in p and q, containing the parameters   $a, l, m$ and $\Lambda$:

\begin{eqnarray}
P&=&(a^2-(p-l)^2)\left(1+\frac{1}{3}(p-l)(p+3l)\Lambda\right), \\
Q&=&a^2-l^2-2mq+q^2-\frac{1}{3}\left[3l^2(a^2-l^2)+(a^2+6l^2)q^2+q^4\right]\Lambda.
\end{eqnarray}
We note that we have chosen $a_0 = 1$ \citep{Griffiths2007}, but there are other choices for the coefficients that will provide simpler expressions for $P$ and $Q$. The locations of the event horizon for this solution are denoted by the roots of $Q(q)$. 
To see if the invariants $Q_1, Q_2$ and $Q_3$ detect the horizons, one could pick particular values for $a,l,m$, and $\Lambda$ to determine the roots of $Q$ and test to see if the  invariants share these roots. The expressions for the $Q$ invariants are very large polynomials in $p$ and $q$, and it is not clear that they can be factorized into irreducible polynomials. Cartan invariants consequently allow for the construction of simpler candidates for detection of the horizon.

We define our null frame:
\begin{eqnarray}
\ell&=&\frac{\sqrt{2}}{2}\sqrt{\frac{Q}{p^2+q^2}}dt-\frac{\sqrt{2}}{2}  p^2\sqrt{\frac{Q}{p^2+q^2}} dr-\frac{\sqrt{2}}{2}\sqrt{\frac{p^2+q^2}{Q}}dq, \\
n&=&\frac{\sqrt{2}}{2}\sqrt{\frac{Q}{p^2+q^2}}dt-\frac{\sqrt{2}}{2}  p^2\sqrt{\frac{Q}{p^2+q^2}} dr+\frac{\sqrt{2}}{2}\sqrt{\frac{p^2+q^2}{Q}}dq, \\
m&=&\frac{\sqrt{2}}{2}\sqrt{\frac{P}{p^2+q^2}}dt+\frac{\sqrt{2}}{2}  q^2\sqrt{\frac{P}{p^2+q^2}} dr-\frac{i\sqrt{2}}{2}\sqrt{\frac{p^2+q^2}{P}}dp, \\
\overline{m}&=&\frac{\sqrt{2}}{2}\sqrt{\frac{P}{p^2+q^2}}dt+\frac{\sqrt{2}}{2}  q^2\sqrt{\frac{P}{p^2+q^2}} dr+\frac{i\sqrt{2}}{2}\sqrt{\frac{p^2+q^2}{P}}dp.
\end{eqnarray}
The only nonzero NP curvature scalars are $\Lambda$ and
\beq
\Psi_{2}= \frac13 \frac{\Lambda a^2 l - 4\Lambda l^3 + 3 i m + 3l}{(p+iq)^3}. 
\eeq
At zeroth order of the Cartan-Karlhede algorithm, we obtain as our Cartan invariants the real and imaginary parts of $\Psi_2$, which are functionally independent, and so $t_0 =2$. The zeroth order isotropy group consists of boost and spins, and so dim $(H_0)$ = 2. 
At the first iteration of the algorithm the components of the covariant derivative of the Weyl spinor are \eqref{NablaPsi0}.  Relative to this coordinate system we have 

\beq & \rho = \mu = - \frac{\sqrt{Q}(iq-p) }{\sqrt{2}(p^2+q^2)^\frac32}, & \label{rhoKNUTAds} \\ 
& \tau = \pi = - \frac{\sqrt{P}(iq-p) }{\sqrt{2}(p^2+q^2)^\frac32}. & \label{tauKNUTAds} \eeq
\noindent and so the boosts and spins have already been fixed to the canonical form, implying dim $(H_1)$ = 0. No new functionally independent invariants appear at first order, so that $t_1 = 2$. It is known already \citep{Aman1984} that the Cartan-Karlhede algorithm concludes at the second iteration, since no new functionally independent invariants appear $t_1 = t_2 = 2$ and dim $(H_1)$ = dim $(H_2)$ = 0.

We would like to compute an extended Cartan invariant that detects the event horizon. As before applying the Bianchi identities to \eqref{NablaPsi0} gives the usual extended Cartan invariant   
\beq \frac{\nabla \Psi_{20'}}{\Psi_2} \propto \rho. \eeq
\noindent Computing the roots of $Q(q)$ for arbitrary $a,l,m$ and $\Lambda$ is not a pleasant task. However, for this extended Cartan invariant we do not need to compute them, as it is clear that the zeros of $\rho$ are exactly the zeros of $Q(q)$. As in the Kerr case, the following extended Cartan invariant

\beq \rho^2 - \tau^2 = \frac{ (Q - P) (iq-p)^2}{2 (p^2 + q^2)^3}, \eeq
\noindent will detect the ergosurface.

Of course, it is possible to implement Theorem \eqref{DPSthrm} to generate an extended Cartan invariant that detects the event horizon, and make the choice of frame irrelevant. Working with the real and imaginary of $\Psi_2$ we find the following invariant:

\beq |d \Psi_2 \wedge d \bar{\Psi}_2|^2 =    -\frac{3^4 Q P|\Psi_2|^2}{2(q^2 + p^2)^{4}}.  \eeq

\noindent  This invariant detects the horizons; however, it also vanishes at $p = l \pm a$ and $p = - \frac{-\Lambda l \pm \sqrt{ 4\Lambda^2 l^2 -3 \Lambda}}{\Lambda}$ as well.

\section{Scalar Polynomial Invariants in Terms of Cartan Invariants in 4D} \label{SPIalaCartan}

For the stationary, asymptotically flat (or (anti) de Sitter) 4D black holes we have considered, the SPIs, $I_1,..., I_7$ from (\ref{IInvariant}) may be expressed in terms of the zeroth order Cartan invariants: 
\beq \Psi_2,~~\bar{\Psi}_2,~~ \Phi_{11}, \nonumber \eeq
\noindent and the nonzero first order extended Cartan invariants: 
\beq & D \Psi_2,~~\Delta \Psi_2,~~\delta \Psi_2,~~ \bar{\delta} \Psi_2,D \Phi_{11},~~\Delta \Phi_{11},~~\delta \Phi_{11},~~ \bar{\delta} \Phi_{11},~~ \rho,~~\pi,~~\tau,~~\mu. & \nonumber \eeq
\noindent This follows by computing the SPIs relative to the coframe determined by the Cartan-Karlhede algorithm. To relate the SPIs to the Cartan invariants we first note that \citep{AbdelqaderLake2015}
$$ I_1+iI_2 = 48 {\Psi_2}^2.$$ 
Since $\Psi_2$ is therefore expressible in terms of SPIs, its gradient $\nabla\Psi_2$
(which for scalars is the same as a covariant derivative) can also be
used to form SPIs. From the definitions \eqref{IInvariant}, it is
immediately obvious that $I_5$, $I_6$ and $I_7$ can be expressed using
$\Psi_2$ and $\nabla\Psi_2$ and their complex conjugates, and the same
follows for $I_3$ and $I_4$ using Page and Shoom's equation (9). Writing
$\nabla A.\nabla B$ for $A_{,\mu}B^{,\mu}$, we find that
\begin{equation}
(96\Psi_2)^2 (\nabla\Psi_2 . \nabla\Psi_2) =
12 \cdot 48(\Psi_2)^2(I_3+iI_4)/5,
\label{LHSPSId}
\end{equation}

\noindent so $I_3$ and $I_4$ are the real and imaginary parts of
$160(\nabla\Psi_2 . \nabla\Psi_2)$. We also find
\begin{eqnarray}
I_5&=& (96)^2[(\Psi_2)^2 (\nabla\Psi_2 . \nabla\Psi_2) + cc
    +2 \Psi_2\bar{\Psi}_2(\nabla\Psi_2 . \nabla\bar\Psi_2)]/4,\\
I_6&=& (96)^2 [-(\Psi_2)^2 (\nabla\Psi_2 . \nabla\Psi_2)
      -cc
    +2 \Psi_2\bar{\Psi}_2(\nabla\Psi_2 . \nabla\bar\Psi_2)]/4,\\
I_7 &=& (96)^2 [(\Psi_2)^2 (\nabla\Psi_2 . \nabla\Psi_2) - cc]/4,
\end{eqnarray}
\noindent where cc means the complex conjugate of the preceding expression. Using the Bianchi identities, these expressions may be simplified

We can now easily compute the $Q_i$ which are
\begin{eqnarray}
Q_1 &=& \frac{  2 {\cal {R}} [({\bar{\Psi}_2}^2(\nabla\Psi_2.\nabla\Psi_2)]}
        {9(\Psi_2\bar{\Psi}_2)^{5/2}},\label{NoteQ1}\\
Q_2 &=&\frac{ -2 ||\nabla{\bar{\Psi}_2}\wedge\nabla\Psi_2||^2}
       {18^2(\Psi_2\bar{\Psi}_2)^3},\\
Q_3 &=& \frac{\nabla\Psi_2 .\nabla\bar{\Psi}_2}
         {18 (\Psi_2\bar{\Psi}_2)^{3/2}},
\end{eqnarray}
where ${\cal R}$ denotes the real part.  We note that while the original formula \eqref{Q} for $Q_2$ is more complicated, it is in fact a dimensionless version of our proposed invariant $W = ||d\Psi_2 \wedge d\bar{\Psi}_2||^2.$

\begin{itemize}

\item {\bf Kerr-NUT-(Anti)-de Sitter Metric} 

For the Kerr-NUT-(Anti)-de Sitter metric we have, using \eqref{BianchiVacuum} and evaluating in the canonical frame:
\begin{eqnarray}
\nabla\Psi_2 . \nabla\Psi_2&=& 18{\Psi_2}^2(\rho^2 - \tau^2),\label{DPDP}\\
\nabla \Psi_2 .\nabla\bar{\Psi}_2 &=& 18\Psi_2\bar{\Psi}_2(|\rho|^2+|\tau|^2).\label{DPDPbar}
\end{eqnarray}

Therefore, the invariants $Q_1, Q_2$ and $Q_3$ take the form:
\begin{eqnarray}
Q_1&=& \frac{ (\rho^2-\tau^2)+cc}{(\Psi_2\bar{\Psi}_2)^{1/2}}, \\
Q_2&=& \frac{2(\rho \bar{\tau} + \bar{\rho} \tau)^2 }{|\Psi_2|^2},\\
Q_3 &=& \frac{ (|\rho|^2 + |\tau|^2) }{|\Psi_2| }.
\end{eqnarray}

\item {\bf Kerr-Newman Metric}

For Kerr-Newman, using \eqref{Bianchi} and evaluating in the canonical frame we find:
\begin{eqnarray}
\nabla\Psi_2 . \nabla\Psi_2&=&  8( \rho^2 - \tau^2 ) \Phi_{11}^2 + 24(\tau^2 + \rho^2) \Psi_2 \Phi_{11} + 18{\Psi_2}^2(\rho^2 - \tau^2),\label{DPDPKN}\\
\nabla \Psi_2 .\nabla\bar{\Psi}_2 &=& 8(|\rho|^2 + |\tau|^2) \Phi_{11}^2 + (12 (|\rho|^2 - |\tau|^2) \bar{\Psi}_2 + cc) \Phi_{11}+18\Psi_2\bar{\Psi}_2(|\rho|^2+|\tau|^2).\label{DPDPbarKN}
\end{eqnarray}

Using these identities, we find that $Q_1$, $Q_2$ and $Q_3$ are now polynomials in terms of $\Phi_{11}$. 

\beq Q_1 = \frac{8}{9} \frac{( {\cal {R}}[ \bar{\Psi}_2^2 ( \rho^2 - \tau^2  ) ]  ) \Phi_{11}^2 }{|\Psi_2|^5} + \frac{8}{3} \frac{{\cal {R}} [ (\rho^2 + \tau^2) \bar{\Psi}_2 ] \Phi_{11} }{|\Psi_2|^4} + \frac{ 2{\cal {R}} [ \rho^2 - \tau^2]  }{|\Psi_2|},  \eeq

\beq Q_2 &=& \frac{4}{81} \frac{( \bar{\tau} \rho + \tau \bar{\rho})^2 \Phi_{11}^4}{|\Psi_2|^6}  + \frac{32}{27} \frac{{\cal {R}} [ \Psi_2( \tau \bar{\rho} - \rho \bar{\tau}) ( \bar{\tau} \rho + \bar{\rho} \tau) \Phi_{11}^3 }{|\Psi_2|^6} \nonumber \\ 
&& + \frac{16}{9} \frac{[|\Psi_2|^2 {\cal R}[ (\rho \bar{\tau})^2] + {\cal R}[ \Psi_2(\bar{\tau}\rho - \bar{\rho} \tau)^2]] \Phi_{11}^2}{|\Psi_2|^6} \nonumber \\
&&+ \frac{8}{3} \frac{(\rho \bar{\tau} + \bar{\rho} \tau )( {\cal R} [ \Psi_2(\bar{\tau} \rho- \bar{\rho} \tau) ]) \Phi_{11}}{|\Psi_2|^4} +  \frac{2(\rho \bar{\tau} + \bar{\rho} \tau)^2 }{|\Psi_2|^2}, \eeq

\beq Q_3 = \frac{4}{9} \frac{(|\rho|^2 + |\tau|^2) \Phi_{11}^2}{|\Psi_2|^3} + \frac{4}{3} \frac{{\cal R} [ \Psi_2 (|\rho|^2 -|\tau|^2)] \Phi_{11} }{|\Psi_2|^3} + \frac{ (|\rho|^2 + |\tau|^2) }{|\Psi_2| }. \eeq

\noindent Due to the $\Phi_{11}$ linear term in $Q_1$, it no longer detects the ergosurface.

\item {\bf Reissner-Nordstr{\"o}m-(Anti)-de Sitter Metric}

This is just a special case of the Kerr-Newman solution where $\tau = \pi = 0$, $\bar{\rho} = \rho$ and $\bar{\Psi}_2 = \Psi_2$, implying that $Q_2 = 0$ and

\beq Q_1 = 2Q_3 = \frac{8 \rho^2 \Phi_{11}^2}{9 \Psi_2^3} + \frac{8\rho^2 \Phi_{11}}{3 \Psi_2^2} + \frac{2\rho^2}{\Psi_2}. \eeq

\noindent Unsurprisingly $Q_1 \propto ||d\Psi_2||^2$ due to \eqref{NoteQ1}. 
\end{itemize}

\section{Examples in 5D}


In this section we will apply the Cartan-Karlhede algorithm to 5D analogues of the black hole metrics studied in the previous section. In particular, we will show how the Cartan invariants are a more viable tool for locating the horizons than the corresponding SPIs generated by Theorem \ref{PSthrm}. While we have not included the invariants, we note that Theorem \ref{DPSthrm} will generate smaller extended Cartan invariants using the non-constant zeroth order Cartan invariants for each of these examples. 

As in the 4D examples, some of these solutions are special cases of the others. For example, the Tangherlini metric is a special case of the Reissner-Nordstr\"{o}m-(Anti)-de Sitter metric. The Tangherlini metric is also a special case of the simply rotating Myers-Perry metric which is in turn a special case of the Kerr-NUT-Anti-de Sitter metric.

\subsection{Tangherlini Metric}
The 5D pseudo-Riemannian analogue of the Schwarzschild metric is given by \cite{Tan63} 
	\begin{equation}
		ds^2=-\left(1-\frac{r_s^2}{r^2}\right)dt^2+\left(1-\frac{r_s^2}{r^2}\right)^{-1}dr^2+r^2\left(d\theta^2+\sin^2{\theta}d\phi^2+\sin^2\theta\sin^2{\phi} d\omega^2\right),
	\end{equation}
where $r_s^2 =2M$ is the Schwarzschild radius. In order to try to detect the event horizon using Cartan invariants, we employ the higher-dimensional analogue of the NP Formalism \citep{ref4}. Defining an orthonormal frame by
\beq
	e_0=\sqrt{1-\frac{r_s^2}{r^2}}dt, \quad e_1=\sqrt{\left(1-\frac{r_s^2}{r^2}\right)^{-1}}dr, \quad  
	 e_2=rd\theta, \quad e_3=r\sin{\theta}d\phi, \quad e_4=r\sin{\theta}\sin{\phi}d\omega,
\eeq
\noindent we produce the half-null frame
\beq
	l=\frac{e_1-e_0}{\sqrt{2}} , \quad n=\frac{e_0+e_1}{\sqrt{2}}, \quad m^{2}=e_2,~~m^{3}=e_3, \quad m^{4}=e_4.&
\eeq


For the zeroth iteration of the Cartan-Karlhede algorithm, we obtain seven nonzero components of the Weyl tensor.
However, these components are algebraically dependent on the following components\footnote{ To display components here and in the following subsections we will repeat indices; this will not indicate summation, unless indicated by a repeated index being raised.} \citep{CHOL2012}:

\beq & C_{0101} = \frac13 C_{0i1i} = -\frac{6 r_s^2}{2r^4},& \eeq

%

\noindent At zeroth order, only null rotations alter the form of the Riemann tensor. The elements of the invariance group at zeroth order $H_0$ consists of rotations and boosts, and hence is four-dimensional. We write the sole linearly independent component as $C_0 = -\frac{r_s}{r^4}$

At first order, the invariance group $H_1$ is the group of spatial rotations specified by three parameters. The one-dimensional subgroup of boosts alters the form of the first covariant derivative of the Riemann tensor, using this we have fixed the boosts by setting the component $C_{0101;1} = 1$: 

\begin{eqnarray} && C_{0101;1} = 3 C_{0i1i;1} = - 3 C_{ijij;1} =3C_{011i;i} = 6C_{1iij;j}= 1,   \\
&&C_{0101;0} = 3C_{0i1i;0} = -3C_{ijij;0} = 3 C_{010i;i} = -6C_{0iij;j} =  \frac{72(r^2-r_s^2)r_s^2}{r^{12}} ,
\end{eqnarray}
\noindent To complete the algorithm one must compute the second covariant derivative of the Weyl tensor, revealing that $t_1=t_2 = 1$ and $dim(H_1)=dim(H_2) = 3$, thus the algorithm stops at second order. 
 

Notice that all positive b.w. terms detect the horizon at first order. In fact, for all higher order derivatives of the Weyl and Ricci tensors, the positive b.w. terms will vanish on the horizon, suggesting that the geometric horizon conjecture for weakly isolated horizons is valid in higher dimensions \citep{ADA2017,AD2017}. At first order the SPI $I_3$ vanishes on the horizon $r=r_s$
\begin{equation}
	I_3=R^{abcd;e}R_{abcd;e}=\frac{2^6 3^3 r_s^4(r_s^2-r^2)}{r^{12}}.
\end{equation}
\noindent Alternatively, since the cohomogeneity of the Tangherlini metric is $n=1$, we can compute the norm of the exterior derivative of $$I_1=R^{abcd}R_{abcd}= \frac{72r_s^4}{r^8},$$
\noindent which yields:
\beq ||\mathrm{d}I_1 ||^2 =\frac{ 2^{12} 3^4 r_s^8 (r_s^2-r^2) }{r^{20}}. \eeq

\subsection{5D Reissner-Nordstr{\"o}m-(Anti)-de Sitter Metric}
From \cite{stability}, the metric is
\beq ds^2 = -f(r) dt^2 + \frac{dr^2}{f(r)} + r^2 d S_3 , \qquad  f(r) = 1- \frac{2M}{r^2} -\frac{\Lambda r^2}{6} + \frac{Q^2}{r^4}, \eeq
\noindent where $dS_3$ is the line element for the unit 3-sphere. 
\noindent We use the following orthonormal frame:
\beq  e_0 = \sqrt{ f(r)} dt, \quad e_1 = \sqrt{\frac{1}{f(r)}} dr, \quad 
 e_2 = r d \theta,~~e_3 = r\sin(\theta) d \phi, \quad e_4 = r\sin(\theta) \sin(\phi) d \omega.   \eeq
\noindent From which we build the half-null frame: 
\beq l= \frac{1}{\sqrt{2}}(e_1 - e_0) , \quad 
 n=\frac{1}{\sqrt{2}}(e_0 + e_1), \quad m_2 = e_2, \quad  m_3 = e_3, \quad  m_4 = e_4.  \eeq
\noindent In this frame, $l$ and $n$ are WANDs; to see this we compute the components of the Weyl and Ricci tensor:
\beq & R_{01} = \frac{2(\Lambda r^6 - 6 Q^2)}{3r^6},~~ R_{ii} = \frac{2(\Lambda r^6 + 3 Q^2)}{3r^6},~~i\in[2,4]&, \eeq
\beq & C_{0101} = 3 C_{0i1i} = \frac{3}{2} \frac{4Mr^2 - 5Q^2}{r^6},& \eeq

\noindent with the remaining nonzero components
$C_{ijij}~~i,j \in [2,4], i \neq j$ algebraically dependent on $C_{0101}$. That is, relative to this frame, the only nonzero components are the b.w. zero terms.

At zeroth order, it can be shown that the isotropy group of the Weyl and Ricci tensor consists of boosts and any spatial rotation; hence dim $(H_0) = 4$. The number of functionally independent invariants is $t_0 =1 $. Continuing the Cartan-Karlhede algorithm, we compute the covariant derivative of the Weyl and Ricci tensor: 
\small
\begin{eqnarray} && R_{01;1} = -4 R_{1i;i} = -2 R_{jj;1}=  \frac{8Q^2}{8Mr^2-15 Q^2},~~R_{01;0} = -4 R_{0i;i} = -2 R_{jj;0} =  -\frac{36(8Mr^2-15Q^2)f(r) Q^2}{r^{14}}, \\
&& C_{0101;1} = 3 C_{0i1i;1} =  -3 C_{ijij;1} = 1,  \\
&&C_{0101;0} = 3C_{0i1i;0} = -3C_{ijij;0} =  -\frac{9}{2} \frac{(8Mr^2 - 15 Q^2)^2 f(r)}{r^{14}} ,\\
&& C_{011i;i} = 2 C_{1iij;j} =  \frac23 \frac{4Mr^2-5Q^2}{8Mr^2-15Q^2},  \\ 
&& C_{100i;i} = 2 C_{0iij;j} =  \frac{3(8Mr^2 - 15 Q^2)  f(r)(4Mr^2-5Q^2)}{r^{14}}.  \end{eqnarray}
\normalsize
\noindent Here we have fixed the boosts by setting the component $C_{0101;1} = 1$. Through direct inspection, it is clear that spatial rotations have no effect on the first order Cartan invariants, hence dim $(H_1) = 3$. The number of functionally independent invariants remains $t_1 =1$. The Cartan-Karlhede algorithm continues for one more iteration since $t_1 = t_2 = 1$ and dim $(H_1)$ = dim $(H_2) = 3$.  

As in the Tangherlini metric, all positive b.w. terms detect the horizon at first order. Since the cohomogeneity is $n=1$, we may produce a SPI that detects the horizon using $I_1 = C_{abcd}C^{abdc}$:

\beq ||\mathrm{d}I_1 ||^2 =\frac{ 2^6  (4Mr^2-5Q^2)^2(8Mr^2-15Q^2)^2f(r)}{r^{26}}. \eeq 
\noindent This invariant will vanish at $r^2 = \frac{5Q^2}{4M}$ and $r^2 = \frac{15 Q^2}{8M}$ as well.

\subsection{5D Simply Rotating Myers-Perry Metric}
The simply rotating Myers-Perry metric is a 5D analogue of the Kerr metric \citep{myeper86,PP2005}. The metric is:
\begin{eqnarray}
\mathrm{d}s^2 &=& -{\frac{1-x}{1-y}}(\mathrm{d}t+R\sqrt{\nu}(1+y)\mathrm{d}\psi)^2+\frac{R^2}{(x-y)^2}[(x-1)((1-y^2)(1-\nu y)\mathrm{d}\psi^2 \\
	&+& \frac{\mathrm{d}y^2}{(1+y)(1-\nu y)})+(1-y)^2(\frac{\mathrm{d}x^2}{(1-x^2)(1-\nu x)} + (1+x)(1-\nu x)\mathrm{d}\phi^2)]. \nonumber
\end{eqnarray}

\noindent  We first define a non-orthogonal half-null frame $\{L_+, L_-, \partial_\phi, \partial_y, \partial_\psi \}$ where:

\beq
L_\pm = \frac{1}{(x^2 -1)(\nu y-1)} \left(\frac{\nu yx-y+\nu x+1-2\nu y}{x-y}R\partial_t -\sqrt{\nu}\partial_\psi \right) \pm \sqrt{\frac{\nu x-1}{(x-y)(y-1)}} \left(\partial_x + \frac{y^2-1}{x^2-1}\partial_y \right).
\eeq

\noindent \cite{PP2005} suggest using a half-null frame $\{l, n, m^2, m^3, m^4\}$ with $l\propto L_+$ and $n\propto L_-$, since in this frame the only nonzero components of the Weyl tensor are those with boost weight zero (i.e. $l$ and $n$ are the WANDs of the simply rotating Myers-Perry metric). We thus start with $\{L_+, L_-, \partial_\phi, \partial_y, \partial_\psi \}$, normalize $L_+$ and $L_-$, and then use the Gram-Schmidt procedure to obtain $\{l, n, m^2, m^3, m^4\}$.

For the zeroth iteration of the Cartan-Karlhede algorithm, we obtain 10 nonzero components of the Weyl tensor.
However, these components are functionally dependent on any two components (say, for example, $C_{0101}$ and $C_{0212}$); thus $t_0 = 2$. As expected, only components with zero boost weight show up, therefore the Weyl tensor is invariant under a boost. Spatial rotations about $m_{4}$ do not change the components of the Weyl tensor. To see why, consider the matrices defined in Table 1 in \cite{CHOL2012}:

\beq C_{0101} = \frac{(x-y)^2(4\nu x+\nu -3)}{4(y-1)^2 R^2}, \eeq

\begin{equation}
M_{ij} = C_{0i1j} =
\begin{pmatrix}
 \frac{1}{4} \frac{(x-y)^2 (2\nu x+\nu -1)}{(y-1)^2 R^2} & -\frac{1}{2} \frac{\sqrt{(1-\nu x)(\nu )(x+1)}(x-y)^2}{(y-1)^2 R^2} & 0 \\
\frac{1}{2} \frac{\sqrt{(1-\nu x)(\nu )(x+1)}(x-y)^2}{(y-1)^2 R^2} & \frac{1}{4} \frac{(x-y)^2 (2\nu x+\nu -1)}{(y-1)^2 R^2} & 0 \\
0 & 0 & -\frac{1}{4} \frac{(x-y)^2 (\nu +1)}{(y-1)^2 R^2} 
\end{pmatrix},
\end{equation}

\begin{equation}
A_{ij} = C_{01ij} =
\begin{pmatrix}
0 & \frac{\sqrt{(1-\nu x)(\nu )(x+1)}(x-y)^2}{(y-1)^2 R^2} & 0 \\
-\frac{\sqrt{(1-\nu x)(\nu )(x+1)}(x-y)^2}{(y-1)^2 R^2} & 0 & 0 \\
0 & 0 & 0
\end{pmatrix}.
\end{equation}


\noindent Since $A_{ij}=\epsilon_{ijk} w^k$, rotations about $m_4$ do not change $A_{ij}$. And from the form of $M_{ij}$, it follows that $M_{ij}$ is unaffected by spatial rotations about $m_4$. Therefore $\dim(H_0)=2$. 

%
%

At first order of the Cartan-Karlhede algorithm, there are several nonzero components of the first covariant derivative of the Weyl tensor. No new functionally independent invariants appear, and the remaining isotropy may be fixed by applying a boost to set $C_{0101;1} = 1$ and $C_{0101;3}=0$; therefore, $t_1 = 2$ and $dim(H_1)=0$. The algorithm would carry on for one more iteration, since $t_1=t_2 = 2$ and dim $(H_1)$ = dim $(H_2) = 0$, we will omit these details. 

We note that the following component at first order detects the horizon which is located at $x=y=1/\nu$:

\begin{eqnarray}
C_{0101;0} &=& \frac98 \frac{(x-y)^5 (2\nu x+\nu-1)^2 (\nu x -1) (\nu y -1) (x-1)}{(y-1)^6 R^6 (\nu+1)}.  \end{eqnarray}

\noindent In fact, all positive b.w. components of the covariant derivative of the Weyl tensor vanish, and similarly for all higher order derivatives. 

%
%
%

As an alternative using SPIs, define $I_1 = C_{abcd}C^{abdc}$ and $J_1 = C_{abcd}C^{abef}C^{cd}_{~~ef}$, and applying Theorem \ref{PSthrm}:
\beq &
||\mathrm{d}I_1 \wedge \mathrm{d}J_1||^2 =
 \frac{-2^4 3^4 \nu^2 (1+y)(x+1)(x-1)^2(1-\nu x)(1-\nu y)(x-y)^{22}(\nu+1)^4(8\nu^2 x^2 + 8\nu^2 x +3\nu^2 -8\nu x -2\nu +3)^2}{(y-1)^{24}R^{24}},\nonumber \eeq

\noindent which will vanish on the horizon.

\subsection{5D Kerr-NUT-Anti-de Sitter Metric}

For the 5D Kerr-NUT-Anti-de Sitter solution, we will use the metric relative to the coordinate system given by equations (22)-(23) in \cite{generalkerr}:

\begin{equation}
\mathrm{d}s^2 = \frac{\mathrm{d}x_1^2}{Q_1}+ \frac{\mathrm{d}x_2^2}{Q_2} + Q_1 \left(\mathrm{d}\psi_0+x_2^2 \mathrm{d}\psi_1\right)^2 +Q_2 \left(\mathrm{d}\psi_0+x_1^2 \mathrm{d}\psi_1\right)^2-\frac{c_0}{x_1^2 x_2^2} \left( \mathrm{d}\psi_0 + \left(x_1^2 + x_2^2 \right) \mathrm{d}\psi_1 + x_1^2 x_2^2 \mathrm{d}\psi_2 \right)^2
\end{equation}

\noindent where 
\beq Q_1 = \frac{X_1}{U}, Q_2 = -\frac{X_2}{U}, U = x_2^2-x_1^2,~~ X_1 = c_1 x_1^2 + c_2 x_1^4 + \frac{c_0}{x_1^2} -2b_1, \text{ and } X_2 = c_1 x_2^2 + c_2 x_2^4 + \frac{c_0}{x_2^2} -2b_2. \eeq

The constants $c_0, c_1, c_2, b_1, b_2$ are free parameters, which are related to the rotation parameters $a_1$, $a_2$, the  mass and NUT charge $M_1, M_2$, and the cosmological constant $\Lambda$ as follows:

\beq
c_0 = a_1^2 a_2^2, \quad c_1 = 1-\frac{\Lambda^2}{4}(a_1^2 + a_2^2) , \quad 
c_2 = \frac{\Lambda}{4} , \quad
 b_\mu = \frac{1}{2} (a_1^2 + a_2^2 - a_1^2 a_2^2 \frac{\Lambda^2}{4}) - M_\mu , \quad \mu = 1, 2.
\eeq

\noindent This metric has been Wick rotated and so it no longer has a Lorentzian signature. This will lead to complex null vectors relative to this coordinate system. However, relative to the original coordinates in \cite{generalkerr} they will be real. 

We first define an orthonormal frame: 
\beq &e_0 = \frac{\mathrm{d}x_1}{\sqrt{Q_1}}, \quad e_1 = \frac{\mathrm{d}x_2}{\sqrt{Q_2}}, &  \\
& e_2 = \sqrt{Q_1} \left(\mathrm{d}\psi_0+x_2^2 \mathrm{d}\psi_1\right), \quad e_3 = \sqrt{Q_2} \left(\mathrm{d}\psi_0+x_1^2 \mathrm{d}\psi_1\right), \quad e_4 = \frac{\sqrt{-c_0}}{x_1 x_2} \left( \mathrm{d}\psi_0 + \left(x_1^2 + x_2^2 \right) \mathrm{d}\psi_1 + x_1^2 x_2^2 \mathrm{d}\psi_2 \right). & \eeq 
\noindent Then, according to \cite{kerrcurvature} and \cite{KNAwands}, the WANDs are simply the null vectors $n$ and $\ell$ in the following half-null frame: 
\beq  l=\frac{i}{\sqrt{2Q_2}} (e_1 +i e_3), \quad n = -i\sqrt{\frac{Q_2}{2}}(e_1 -i e_3), \quad  
 m_2 = e_0, \quad m_3 = e_2, \quad m_4 = e_4.   \eeq

Using the WANDs in this half-null frame, it may be shown that any of the components are functionally dependent on the choice of two components at zeroth order. Thus $t_0 = 2$. All components are of b.w. zero and they do not change under a rotation about $m_4$. To see why, we express the Weyl tensor components as the following matrices as defined by Table 1 in \citep{CHOL2012}:

\beq C_{0101} &= -\frac{2(x_1^2 + 3x_2^2)(b_1-b_2)}{U^3}, \eeq 

\begin{equation}
M_{ij} = C_{0i1j} =
\begin{pmatrix}
-\frac{2(x_1^2+x_2^2)(b_1-b_2)}{U^3} & \frac{4ix_1 x_2 (b_1-b_2)}{U^3} & 0 \\
-\frac{4ix_1 x_2 (b_1-b_2)}{U^3} & -\frac{2(x_1^2+x_2^2)(b_1-b_2)}{U^3} & 0 \\
0 & 0 & -\frac{-2(b_1-b_2)}{U^2}
\end{pmatrix},
\end{equation}

\begin{equation}
A_{ij} = C_{01ij} =
\begin{pmatrix}
0 & \frac{	8ix_1 x_2 (b_1-b_2)}{U^3} & 0 \\
-\frac{8ix_1 x_2 (b_1-b_2)}{U^3} & 0 & 0 \\
0 & 0 & 0
\end{pmatrix}.
\end{equation}

\noindent We note that this is a vacuum solution and so $R_{ij} =
\Lambda g_{ij}$. Since $A_{ij}=\epsilon_{ijk} w^k$, rotations about
$m_4$ do not change $A_{ij}$. And from the form of $M_{ij}$, it follows that $M_{ij}$ is
  unaffected by spatial rotations about $m_4$. Thus
  $\dim(H_0)=2$.

At first iteration, we have several non-trivial components, but they
are all functionally dependent on the two functionally independent
invariants at zeroth order, $t_1 = 2$. We can fix the remaining
isotropy by applying a boost to set $C_{0101;1} =  1$: a rotation about $m_4$ is not needed as
our frame already gives the canonical choice $C_{0101;3} = 0$. Therefore, $\dim(H_1)=0$. The algorithm would carry on for one more iteration, since $t_1=t_2 = 2$ and dim $(H_1)$ = dim $(H_2) = 0$; however, we will omit these details. Instead of listing components of the covariant derivative of the Weyl tensor, we note that the following components at first order detect the horizon, which occurs when the function $X_2 = 0$:
\begin{equation}
C_{0101;0}=  \frac{2^5\cdot 3^2 (x_1^2+x_2^2)^2(b_1-b_2)^2 x_2^2 Q_2}{U^8}
\end{equation}

To determine the location of the event horizon, we may compute the
expansion of the trivially-boosted\footnote{The boost parameter used was chosen to simplify calculations. We note that this parameter is not well-defined for the original purpose of the Cartan-Karlhede algorithm}  $\ell$ \citep{KNAwands}, $\theta_{(\ell)} = \frac{1}{3} h^{ab} \ell_{(a;b)} $, where $h_{ab} = g_{ab} - \ell_{(a} n_{b)}$: 
\beq \theta_{(\ell)} =   \frac{4(x_1^2-3x_2^2)(x_1^2+x_2^2) Q_2(b_1-b_2)}{U^5}. \label{KNAexpansion} \eeq 
As in the previous example, all of the positive b.w. components of the covariant derivative of the Weyl tensor vanish on the horizon, and similarly for all higher order derivatives of the Weyl tensor. Applying Theorem \ref{PSthrm}, we may produce a SPI that detects the horizon using $I_1 = C_{abcd}C^{abdc}$ and $J_1 = C_{abcd}C^{abef}C^{cd}_{~~~ef}$:
\beq
||\mathrm{d}I_1 \wedge \mathrm{d}J_1||^2 = \frac{ 2^{37} 3^{4}  
 (3x_1^4 +2x_1^2x_2^2+3x_2^4)^2 x_1^2 x_2^2 X_1 X_2
 (b_1-b_2)^{10}}{(x_1-x_2)^{30} (x_1+x_2)^{30}}.
\eeq
\noindent Alternatively, we could use Theorem \ref{DPSthrm} to produce an
extended Cartan invariant from the non-constant zeroth order Cartan
invariants that will detect the horizon and will be of lower order
than the above SPI.

\section{Conclusion}

We have shown that it is possible to locate the event horizon of any stationary, asymptotically flat (or (anti) de Sitter) black hole from the zeros of Cartan invariants. Our work complements the related results on the detection of stationary horizons using SPIs \citep{AbdelqaderLake2015,PageShoom2015}. Our approach has a notable advantage in that it is computationally less expensive compared to the related SPIs. In the reviewed examples we have also computed extended Cartan invariants whose zeros only occur on the surface of the stationary horizons, and the related SPIs \citep{PageShoom2015} are computed for each solution as a comparison. In  4D, we employ the NP formalism relative to the frame arising from the Cartan-Karlhede algorithm to demonstrate the relationship between the SPIs and the Cartan invariants. 

{While we have only considered stationary horizons with spherical topology, in higher dimensions other topologies are permitted for the horizon. For example, the 5D black rings have horizon topology $S^1 \times S^2$. For the rotating and supersymmetric black rings, it has been shown that the approach based on Cartan invariants will detect the horizon \citep{blackring}. Furthermore, the results of \cite{blackring} show that the Cartan-Karlhede algorithm can be implemented to produce Cartan invariants that detect the horizon even when WANDs are not known. This indicates that the Cartan-Karlhede algorithm can be implemented in dimensions $D \geq 5$ and that the resulting invariants will be easier to compute than the related SPIs.} 

In future work we will consider the horizons of solutions containing more than one black hole, including the analytical example of the Kastor-Traschen solution \citep{kastor}. This dynamical extension may allow us to follow the formation of the event horizon during the merger of two black holes, during the  phase of collapse of a star into a single black hole \citep{collapse}, and perhaps even the disappearance of the horizon during the evaporation of a single black hole  \citep{evaporation}.  We will also extend our method to the study of evolving event horizons for time dependent metrics, including metrics currently used for cosmological modelling. We hope that these results will play an important role in numerical relativity in which configurations of many black holes are evolved in time \citep{shapiro}, and a sharp localization of the event horizons is required.

\section*{Acknowledgements}  
The authors would like to thank Jan {\r A}man and Sebastian Jan Szybka for checking the calculations in the paper.  
The work was supported by NSERC of Canada (A.A.C, A.F., D.B.)  and
AARMS (D.G.) and through the Research Council of Norway, Toppforsk
grant no. 250367: Pseudo-Riemannian Geometry and Polynomial Curvature
Invariants: Classification, Characterisation and Applications
(D.M.). We are also grateful to the authors of the free software
  GRtensorII, Reduce and Sheep used in checking our calculations.

\bibliographystyle{plainnat}
\bibliography{GRstrings,long_v4}

\end{document}